\documentclass[11pt]{article}

\usepackage[compact]{titlesec}
\usepackage{times}
\usepackage{latexsym}
\usepackage{amsfonts,amsthm,amssymb}
\usepackage{amsmath}
\usepackage{euscript}
\usepackage{amstext}
\usepackage{graphicx}
\usepackage{color}
\usepackage{url,hyperref}
\usepackage{verbatim}
\usepackage{xspace}
\usepackage{framed}
\usepackage[boxed]{algorithm}
\usepackage{algorithmic}
\usepackage{enumitem}
\usepackage{bbm}

\newtheorem{lemma}{Lemma}[section]
\newtheorem{theorem}[lemma]{Theorem}

\newtheorem{corollary}[lemma]{Corollary}

\newcommand{\eps}{\epsilon}

\newcommand{\etal}{et al.\xspace}

\newcommand{\cF}{\mathcal{F}}

\newcommand{\cI}{\mathcal{I}}

\newcommand{\lpp}{\mathsf{LP}_\mathsf{primal}}
\newcommand{\lpd}{\mathsf{LP}_\mathsf{dual}}

\newcommand{\initOneLiners}{%
    \setlength{\itemsep}{0pt}
    \setlength{\parsep }{0pt}
    \setlength{\topsep }{0pt}
}

\newcommand{\cost}{\textsc{cost}}

\begin{document}
\sloppy

\title{
A Tight Approximation for Co-flow Scheduling for Minimizing Total Weighted Completion Time}

\author{Sungjin Im\thanks{ Electrical Engineering and Computer Science, University of California, 5200 N. Lake Road, Merced CA 95344. {\tt sim3@ucmerced.edu}. }  \and Manish Purohit\thanks{Google, Mountain View, CA} 
}

\date{}
\maketitle
\thispagestyle{empty}

\begin{abstract}
Co-flows model a modern scheduling setting that is commonly found in a variety of applications in distributed and cloud computing. In co-flow scheduling, there are $m$ input ports and $m$ output ports. Each co-flow $j \in J$ can be represented by a bipartite graph between the input and output ports, where each edge $(i,o)$ with demand $d_{i,o}^j$ means that $d_{i,o}^j$ units of packets must be delivered from port $i$ to port $o$. To complete co-flow $j$, we must satisfy all of its demands. Due to capacity constraints, a port can only transmit (or receive) one unit of data in unit time. A feasible schedule at each time $t$ must therefore be a bipartite matching. 

We consider co-flow scheduling and seek to optimize the popular objective of total weighted completion time. Our main result is a $(2+\eps)$-approximation for this problem, which is essentially tight, as the problem is hard to approximate within a factor of $(2 - \eps)$. This improves upon the previous best known 4-approximation. Further, our result holds even when jobs have release times without any loss in the approximation guarantee. The key idea of our approach is to construct a continuous-time schedule using a configuration linear program and interpret each job's completion time therein as the job's deadline. The continuous-time schedule serves as a witness schedule meeting the discovered deadlines, which allows us to reduce the problem to a deadline-constrained scheduling problem. 

\bigskip
\noindent
\textbf{Authors' note:} \emph{ This paper has a bug. The bug is in Section 4.2 Finding a Feasible Integral Schedule Meeting Deadlines -- an integral flow to the flow network we created is not necessarily a valid schedule of the given co-flows. In fact, the following problem is known to be NP-hard: Suppose we are given a bipartite graph where each edge has a certain deadline. At each time, we can schedule a subset of edges if they form a matching. An edge completes when it appears in a matching. The goal is to determine if there exists a feasible schedule that completes all edges before their deadline. Since we already lose a factor two in the approximation ratio in the first step, our approach cannot yield a 2-approximation. However, it is plausible that one can get a better than 4-approximation by finding an alternative and correct rounding that replaces our second rounding in Section 4.2.} 

\emph{We note that although the separation oracle has an similar issue, it is not a big deal, since we can simply use $y_{j,t}$, instead of $y_{j}^F$, which has the meaning that co-flow $j$ completes at time $t$ if $y_{j, t} =1$. }

\emph{We plan to keep this arxiv paper accessible for a while, as some colleagues showed interests in this approach despites the above flaw. We hope that our approach, albeit flawed, can help find a better than 4-approximation for the problem. This manuscript has not been published anywhere.}

\end{abstract}

\clearpage
\setcounter{page}{1}

\section{Introduction}
	\label{sec:intro}

Co-flow scheduling~\cite{chowdhury2012coflow} is an elegant scheduling model that abstracts a signature scheduling problem that characterizes modern massively parallel computing platforms such as MapReduce \cite{dean2008mapreduce} and Spark \cite{zaharia2010spark}. Such platforms have a unique processing pattern that interleaves local computation with communication across machines. Due to the size of the large data sets processed, communication typically tends to be the main bottleneck in the system performance of these platforms, and the co-flow model captures the key scheduling challenge arising from such communication.

In co-flow scheduling, there are $m$ input ports and $m$ output ports. Each job (or co-flow) $j$ is described by a bipartite graph between input and output ports where each edge $(i, o)$ is associated with a demand $d_{i,o}^j$, meaning the job has $d_{i,o}^j$ unit-sized packets to send from input port $i$ to output port $j$. Job $j$ completes at the earliest time when all its packets are delivered. We assume that time is slotted and at each integer time $t$, a feasible schedule is a matching, and exactly one packet of a certain co-flow is delivered from input port $i$ to output port $o$ when $i$ is matched to $o$. At a high-level, a packet delivery is data migration from a machine $i$ to another machine $o$ before starting the next round of local computation. Co-flow scheduling can also be viewed as a generalization of the classical concurrent open shop scheduling problem~\cite{bansal2010inapproximability,chen2007supply,garg2007order,leung2007scheduling,mastrolilli2010minimizing,sachdeva2013optimal,wang2007customer}, as we can recover the concurrent open shop when $d_{i,o}^j = 0$ for all jobs $j$, and pairs $i \neq o$ of ports. 

Since co-flow scheduling is an abstraction of one-round communication which is potentially part of multiple rounds of computation, some packets become ready to be transferred only later in the schedule. This necessitates the study of co-flow scheduling with packet release times. Thus, a job $j$ may have a release/arrival time, $r_j$, meaning that no packet of job $j$ can be transported before its release time. 
More generally, we can consider a different release time for each packet of a job and all our results extend naturally to this setting. For simplicity, we assume that all packets of job $j$ have the same release time $r_j$.

In this paper we study co-flow scheduling to minimize the (weighted) average\ /\ total completion time, one of the most popular objectives used to measure average job latency. Not surprisingly, being at the heart of modern parallel computing, co-flow scheduling has been actively studied by both the system and theory communities \cite{ahmadi2017scheduling,chowdhury2012coflow,chowdhury2015efficient,chowdhury2014efficient,khuller2016brief,luo2016towards,qiu2015minimizing,JahanjouKR16,shafiee2017improved,zhao2015rapier}, particularly for the completion time objective \cite{ahmadi2017scheduling,JahanjouKR16,khuller2017select,khuller2016brief,luo2016towards,qiu2015minimizing,shafiee2017improved}. Since even the special case of concurrent open shop scheduling is hard to approximate within factor $2 - \eps$ \cite{bansal2010inapproximability,sachdeva2013optimal}, there has been a sequence of attempts to give approximate scheduling algorithms. Prior to this work, the best known approximation ratio were 5 for the general case, and 4 when all jobs arrive at time 0 \cite{ahmadi2017scheduling,shafiee2017improved}, respectively.
Thus, it remained open to reduce the gap between best upper and lower bounds. 

\subsection{Our Result}

In this paper, we answer the open problem in the affirmative by giving an essentially tight approximation algorithm.

\begin{theorem}
\label{thm:main}
	There is a deterministic 2-approximation algorithm for co-flow scheduling for minimizing total weighted completion time when all parameters such as demands and release times are polynomially bounded by the input size. For the general case, there is a $(2+\eps)$-approximation for any constant $\eps >0$.
\end{theorem}

This approximation guarantee is almost tight as there is an inapproximability result of $2- \eps$~\cite{bansal2010inapproximability,sachdeva2013optimal}. 
Thus, our result essentially closes the approximability of co-flow scheduling for the completion time minimization objective. 
Interestingly, we have the same approximation guarantee for  general co-flow scheduling with arbitrary release times as for the special case when all jobs are released at time 0. 

\medskip
\noindent
\emph{Other Extensions.} As mentioned before, our result can be easily extended to handle the case where packets of the same job have different release times. As pointed out in the seminal paper \cite{chowdhury2012coflow}, in practice, ports can have non-uniform capacities, meaning that each port can route a different number of packets in each time slot. Our algorithm easily generalizes to handle ports with non-uniform capacities (see Section~\ref{sec:extensions} for a sketch).
\medskip

Interestingly, even for the classical concurrent open shop scheduling with release times, the previous best known approximation factor was 3~\cite{ahmadi2017scheduling,garg2007order,leung2007scheduling}.
Theorem \ref{thm:main} immediately yields an improved approximation algorithm for \emph{preemptive} concurrent open shop with arbitrary release times.

\begin{corollary}
\label{cor:concurrent}
There is a deterministic 2-approximation algorithm for \emph{preemptive} concurrent open shop scheduling with release times when all input parameters such as processing times and release times are polynomially bounded. For the general case, there is a $(2 + \eps)$-approximation algorithm for any constant $\eps > 0$.

\end{corollary}

\subsection{Our Techniques}
    \label{sec:techniques}
We first discuss approaches used in the previous work \cite{qiu2015minimizing,ahmadi2017scheduling}\footnote{The result by Luo \etal \cite{luo2016towards} is unfortunately flawed as pointed out in \cite{ahmadi2017scheduling}.} along with their limitations. We then give a high-level overview of our approach and highlight our key ideas that enable an essentially tight approximation. To streamline our discussion, let's assume for a while that all jobs arrive at time 0. 
    
\paragraph{Previous Approaches.} We first discuss the work by Qiu \etal~\cite{qiu2015minimizing}, which gave the first non-trivial approximation for the completion time objective. A key observation therein is that, individually, each job $j$ can be scheduled within $\Delta(j)$ time units if there were no other jobs where $\Delta(j)$ denotes the maximum degree of the bipartite graph representing job $j$. 
This idea is easily extended to multiple jobs by aggregating them into one job. More precisely, a subset $\mathcal{S}$ of jobs can be completed  within $\Delta(\mathcal{S})$ time steps where $\Delta(\mathcal{S})$ is the maximum degree of the bipartite graph that adds up all demands of each edge over the jobs in $\mathcal{S}$. 
They set up an interval-indexed linear programming relaxation to find a tentative completion time for each job. 
While rounding, jobs with similar tentative completion times are grouped together and scheduled separately in $\Delta(S)$ time steps (where $S$ denotes the set of jobs grouped together). The disjoint schedules for each group are then concatenated sequentially. Since jobs in a group $S$ are scheduled independently of other groups so that all jobs finish by time $\Delta(S)$, their approach can be viewed as a \emph{reduction to the makespan objective} where jobs have a \emph{uniform deadline}. Unfortunately, such an approach cannot give a tight approximation as a result of concatenation of disjoint schedules. Using these techniques, Qiu \etal\ give a randomized 16.54-approximation algorithm.

The work by Ahmadi \etal~\cite{ahmadi2017scheduling} giving the best known 4-approximation improves upon Qiu \etal's result by a \emph{reduction to the Concurrent Open Shop problem} (COS). In the COS problem, there  are $m$ machines. Each job has a certain workload on each machine and completes when all its workload has been completed. Thus, the COS does not have to cope with 
the complexities coming from the bipartite graph structure as in co-flow scheduling. Their LP relaxation adapts the LP relaxation for COS by pretending that each port is an individual machine. Their approach loses a factor of 2 in the rounding just as the LP rounding for the COS problem does. An additional factor 2 loss follows since the same rounding for the COS does not yield a feasible schedule for co-flow as the LP was set up without fully capturing the interaction between input and output ports. 
Their approach crucially uses \emph{non-uniform tentative completion times} for the improvement upon the result by Qiu \etal~\cite{qiu2015minimizing}, but fails to give a tight approximation by \emph{failing to effectively capture the underlying graph structure}.

\paragraph{Our Approach.} To get a tight approximation, we seek to use non-uniform deadlines as in \cite{ahmadi2017scheduling} but at the same time ensure that we fully factor in the underlying graph structure. At a high level, our approach can be summarized as \emph{a true reduction of the problem of minimizing total completion time to a deadline-constrained scheduling problem} in the following sense: We find a tentative completion time / deadline $C^*_j$ for each job $j$ such that \emph{(i) it is guaranteed that there exists a schedule where every job $j$ completes before its deadline $C^*_j$, and (ii) any deadline-meeting schedule is 2-approximate for the total completion time objective}. 

Towards this end, we use a \emph{configuration LP} where we create a variable for each job's each possible schedule pretending that no other jobs exist; note that each variable completely determines the job's schedule, particularly its completion time. The configuration LP has exponentially many variables but can be solved by solving its dual using a separation oracle. The separation oracle boils down to finding the cheapest schedule for a given job when ports are priced differently over time, which can be solved using network flow. The configuration LP gives a fractional solution with a nice property: when a job completes by $\lambda$ fraction by time $t$, $\lambda$ fraction of the job's complete schedules are packed by that time $t$. Using this property, one can find a fractional schedule with deadlines $C^*_j$ giving the desired properties, (i) and (ii)---in particular, unit fraction of each job's schedules is packed by $C^*_j$, which implies a witness fractional flow guaranteeing (i). Then, we again use network flow to find an actual schedule.  

Our approach is readily extended to factor in job release times without sacrificing the approximation quality. In fact, thanks to our simple and general approach,  we can easily extend our result to capture different release times even at the packet level as well as to capture non-uniform port capacities. We believe our high-level approach that determines `tight' non-uniform deadlines for jobs admitting a feasible schedule by constructing a witness fractional schedule, could find more applications in other scheduling contexts.

\subsection{Other Related Work}

As mentioned before, the Concurrent Open Shop problem (COS) \cite{bansal2010inapproximability,sachdeva2013optimal,garg2007order,wang2007customer,mastrolilli2010minimizing,chen2007supply,leung2007scheduling} is a special case of the co-flow scheduling problem. However, in the COS, preemption is typically disallowed, meaning that any task started on a machine must be processed until its completion without interruption once it gets started. It is easy to see that preemption does not help if all jobs arrive at time 0, in which case several 2-approximation algorithms were shown  \cite{chen2007supply,garg2007order,leung2007scheduling} via LP rounding, which were later shown to be tight \cite{bansal2010inapproximability,sachdeva2013optimal}. When jobs have different release times, the same LP relaxations yielded 3-approximations~\cite{garg2007order,leung2007scheduling}. Later, Mastrolilli \etal~\cite{mastrolilli2010minimizing} gave a simple greedy algorithm that matches the best approximation ratio when all jobs arrive at time 0. Recently, Ahmadi \etal \cite{ahmadi2017scheduling} gave a combinatorial 3-approximation via a primal-dual analysis when jobs have non-uniform release times.

\section{Problem Definition and Notation}
	\label{sec:prelim}
    
    Following the approach taken by all prior work~\cite{ahmadi2017scheduling,chowdhury2014efficient,qiu2015minimizing}, we abstract out the network as a single $m \times m$ non-blocking switch with unit capacity constraints, i.e., any input (or output) port can only transmit (or receive) one unit of data at any time slot. Let $m$ denote the number of input and output ports in the system and let $J$ denote the set of all co-flows. As jobs are the more commonly used terminology in scheduling literature, we may refer to co-flows as jobs. Each job $j \in J$ is represented as a bipartite graph $G^j$ between the input and output ports where each edge $(i,o)$ has a weight $d_{io}^j$ that represents the number of unit sized packets that co-flow $j$ needs to deliver from port $i$ to port $o$. 
Each job $j$ also has a weight $w_j$ that indicates its relative importance and a release time $r_j$. A co-flow $j$ is available to be scheduled at its release time $r_j$ and is said to be completed 
when all the flows of the job $j$ have been scheduled. More formally, the completion time $C_j$ of co-flow $j$ is defined as the earliest time such that for every input port $i$
and output port $o$, $d^j_{io}$ units of its data have been transferred from port $i$ to port $o$. We assume that time is slotted and data transfer within the network is instantaneous. Time slots are indexed by positive integers. By time slot $t$, we refer to time interval $[t-1, t)$ of unit length. Since each input port $i$ can transmit at most one unit
of data and each output port $o$ can receive at most one unit of data in each time slot, a feasible schedule for a single time slot can be described as a matching. Our goal is to find a feasible schedule that minimizes the total, weighted completion time of the co-flows, i.e. $\text{minimize} \sum_{j} w_jC_j$.
\section{Linear Programming}
	\label{sec:lp}

\subsection{A Simplifying Assumption}
\label{sec:simplifying}

In this section, we present our configuration linear program. To make our presentation more transparent, we will first assume that the maximum potential completion time $T$ of all jobs is polynomially bounded by the input size. Note that this simplifying assumption holds true when all parameters, particularly job demands and release times are polynomially bounded. The main challenges still remain the same under this assumption, which will be removed in  Section~\ref{sec:remove}. 

\subsection{LP Formulation}

We formulate a configuration linear program. For a job $j$, a configuration is a complete feasible, integral schedule for $j$. More formally, a configuration $F$ for job $j$ is a ternary relation of tuples of the form $(i,o,t)$ that indicates one unit of data is transferred from input port $i$ to output port $o$ at time $t$. Since $F$ is a feasible configuration for $j$, by definition, the relation must satisfy:

\begin{enumerate}
	\item $|\{(i,o,t) \in F\}| = d^j_{io}$ for all $i, o \in [m]$.
	\item  $|\{(i,o,t) \in F\}| \leq 1$ for each pair of $i \in [m]$ and $t \geq 1$. 
	\item  $|\{(i,o,t) \in F\}| \leq 1$ for each pair of $o \in [m]$ and $t \geq 1$. 
\end{enumerate}

Let $\mathcal{F}(j)$ denote the set of all possible feasible schedules of job $j$. Let $C_j^F$ denote the completion time of job $j$ under schedule $F \in \mathcal{F}(j)$. In other words $C_j^F = \max \{t \ | \ \exists i \in [m], o \in [m]\ s.t.\ (i,o,t) \in F\}$. We introduce variables $\{y_j^F\}$ to indicate whether job $j$ is scheduled as per configuration $F$. The following shows our linear programming formulation, which we refer to as $\lpp$.

\begin{align}
  &{\bf \lpp:}&&\hspace{-1cm}\min \sum_{j \in J} w_j \sum_{F \in \mathcal{F}(j)} C_j^F y_j^F    \nonumber\\
\text{subject to, }\quad 
  & \forall j \in J, && \displaystyle \sum_{F \in \mathcal{F}(j)} y_j^F \geq 1  \label{eqn:LPassignment}\\
  & \forall i \in [m] \textnormal{ and } \forall t \in [T], && \displaystyle \sum_{j \in J, o \in [m]} \sum_{\substack{F \in \mathcal{F}(j) \\  s.t. (i,o,t) \in F}} y_j^F \leq 1 \label{eqn:LPinputconflict}\\
  & \forall o \in [m] \textnormal{ and } \forall t \in [T], && \displaystyle \sum_{j \in J, i \in [m]} \sum_{\substack{F \in \mathcal{F}(j) \\ s.t. (i,o,t) \in F}} y_j^F \leq 1 \label{eqn:LPoutputconflict} \\
 & \forall j \in J \textnormal{ and } \forall F \in \mathcal{F}(j), && \displaystyle y_j^F \geq 0 \label{eqn:LPrelax}
\end{align}

In the objective, each job $j$'s weighted completion time is $w_j C^F_j$ under configuration $F \in \cF(j)$.
Constraint \eqref{eqn:LPassignment} ensures that every co-flow is scheduled as per some configuration. Constraints \eqref{eqn:LPinputconflict} and \eqref{eqn:LPoutputconflict} ensure that at any time step $t$ at most one co-flow uses an input port $i$ or output port $o$. In constraint \eqref{eqn:LPrelax}, we relax the integrality requirements to get a linear programming relaxation. 
\subsection{Solving the Linear Program}
	\label{sec:oracle}
	
The linear programming relaxation $\lpp$ has exponentially many variables but polynomially many constraints (under the simplifying assumption that $T$ is polynomially bounded). To solve $\lpp$, we formulate its dual, $\lpd$, which has polynomially many variables but exponentially many constraints, and demonstrate a polynomial-time separation oracle for $\lpd$. Using the Ellipsoid method with the separation oracle, we can solve $\lpd$ in polynomial time. During the execution of the Ellipsoid method, only polynomially many constraints of the dual are considered, thus, as a result of the strong duality theorem, we only need to consider the corresponding variables to optimally solve $\lpp$. Therefore, a polynomial-time separation oracle for $\lpd$ implies that $\lpp$ can be solved in polynomial time.

	The following shows the dual of $\lpp$. 

\begin{align}
  &{\bf \lpd:}&&\hspace{-1cm}\max \sum_{j \in J} \alpha_j  - \sum_{i \in [m], t \in [T]} \beta_{i,t} - \sum_{o \in [m], t \in [T]} \gamma_{o,t}  \nonumber  \\
\text{subject to, }\quad 
  & \forall j \in J \textnormal{ and } \forall F \in \mathcal{F}(j), && \displaystyle \alpha_j - \sum_{(i,o,t) \in F}  (\beta_{i,t} + \gamma_{o,t}) \leq w_j C_j^F  \label{eqn:LPdual}\\
 &  && \displaystyle \alpha, \beta, \gamma \geq 0 \nonumber \label{eqn:LPdualpositive} 
\end{align}

It now remains to show a separation oracle for $\lpd$. Our task is, given $\alpha_j, \beta_{i,t} \textnormal{ and } \gamma_{o,t}$ values, to find a job $j$ and a configuration $F \in \mathcal{F}(j)$ that violates constraints \eqref{eqn:LPdual} or certify that no such constraint exists. We observe that such a separation oracle is reduced to solving the following key problem.

\medskip
\noindent 
\textbf{Key Problem for the Separation Oracle:} 
Given a single job $j$ and a deadline $C_j \in [T]$, find a feasible, integral schedule for $j$ of minimum cost when input port $i$ is priced at $\beta_{i,t}$ at time $t$ and output port $o$ is priced at $\gamma_{o,t}$ at the same time.

\smallskip
The observation immediately follows by considering a fixed job $j$ and a fixed completion time $C_j \in [T]$\ ---\ if the cost of the job's minimum-cost deadline-constrained schedule $F^*$ is less expensive than $\alpha_j -  w_j C_j^{F^*}$, then we report the corresponding configuration as a violated constraint. Otherwise, if the minimum cost schedules for all deadlines satisfy the constraints, then
we are guaranteed that no constraints are violated. Suppose for contradiction that some configuration $F$ violates constraint \eqref{eqn:LPdual}. Let $F^*$ be the minimum cost schedule with deadline $C_j^F$. Thus, we have
$cost(F^*) \leq cost(F) < \alpha_j - w_j C_j^F \leq \alpha_j - w_j C_j^{F^*}$, which is a contradiction.

We show how to reduce this key problem to an instance of the classical minimum cost flow problem. Create a directed graph $G$ as follows. For each $t \in [C_j] \setminus [r_j]$, create a complete, bipartite, directed subgraph $G_t = (I_t, O_t, A_t)$ where part $I_t$ has $m$ nodes corresponding to the input ports, $O_t$ has $m$ nodes corresponding to the output ports. For an edge $e = (i_t, o_t)$, assign a cost of $\beta_{i,t} + \gamma_{o,t}$.
Assign a capacity of one for all \emph{edges and all vertices}\footnote{Vertex capacities can be modeled by splitting the a vertex into two and adding a unit capacity edge connecting the split parts.} in $G_t$. 
In addition, for every pair $i,o \in [m]$, create vertices $source_{io}$ and $sink_{io}$. Connect $source_{io}$ to vertices $i_t$ for all $t \in [C_j] \setminus [r_j]$ and connect vertices $o_t$ for all $t \in [C_j] \setminus [r_j]$ to $sink_{io}$ with unit capacity and zero cost arcs. Finally, add two new vertices $source$ and $sink$. For every pair $i,o \in [m]$, add an arc of capacity $d_{io}^j$ and cost zero from $source$ to $source_{io}$ and also from $sink_{io}$ to $sink$. See Figure \ref{fig:oracle} for an illustration of this construction. 

\begin{figure}[htbp]
\includegraphics[width=\textwidth]{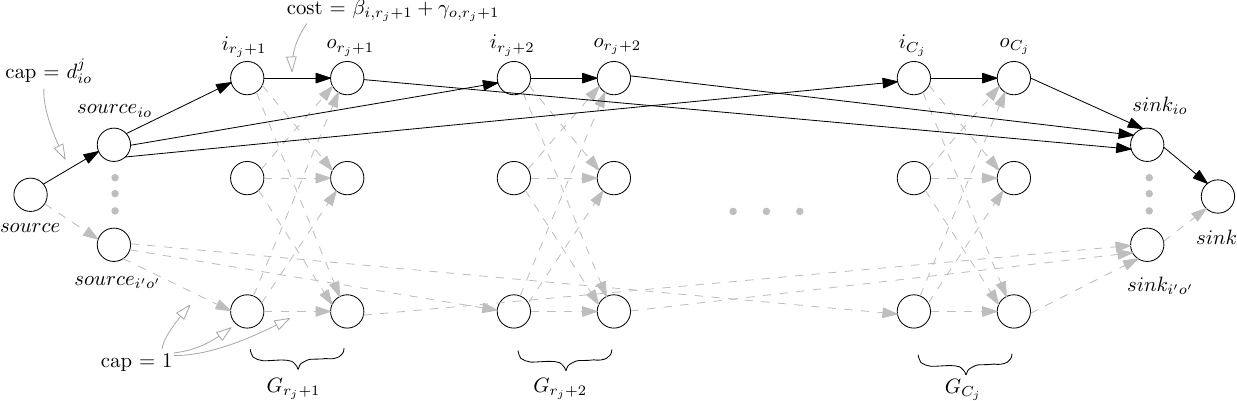}
\caption{An example of graph construction for the minimum cost flow problem}
\label{fig:oracle}
\end{figure}

Note that any integral maximum flow $f$ in $G$ corresponds to a feasible schedule $F$ for job $j$ with completion time at most $C_j$ and vice versa. 
In particular, for any pair $i,o \in [m]$, the maximum flow sends $d_{io}^j$ units of flow from $source_{io}$ to $sink_{io}$ through $d_{io}^j$ edge disjoint paths each of the form $(source_{io} \rightarrow i_t \rightarrow o_t \rightarrow sink_{io})$. The vertex capacities ensure that for any $t \in [C_j] \setminus [r_j]$, the edges of $G_t$ carrying flow form a matching. Thus the set of all triples $(i,o,t)$ such that the edge $(i_t, o_t)$ carries unit flow forms a feasible configuration (schedule) for $j$; job's release time is respected, as only times in $[C_j] \setminus [r_j]$ can be used by $j$.   In addition, the cost of the flow $f$ is exactly equal to $\sum_{(i,o,t) \in F} (\beta_{i,t} + \gamma_{o,t}$). Thus, we can find the feasible schedule $F^* \in \mathcal{F}(j)$ with completion time at most $C_j$ that minimizes $\sum_{(i,o,t) \in F^*} (\beta_{i,t} + \gamma_{o,t}$) by solving the minimum cost maximum flow problem on $G$. 

\section{Rounding}

In this section, we show how to round an optimal feasible solution to $\lpp$ to obtain a 2-approximate integral schedule. Towards this end, we consider some intermediate schedules that are \emph{continuous}. We say that a schedule is continuous if it schedules an \emph{integral} matching at each \emph{instantaneous} time. Thus, if the matching at time $t$ includes edge $(i,o)$ for job $j$, it means that a $dt$ amount of packet of job $j$ is sent from input port $i$ to output port $o$ during an infinitesimally small interval $[t, t + dt)$; or equivalently, the same amount of demand is served for job $j$'s edge $(i,o)$. In a continuous schedule, job $j$ completes at the earliest time when all its demands are served. 

Our rounding scheme consists of two main components.  In Section~\ref{sec:rounding1}, we first find a deadline $C^*_j$ for each job $j$ that admits a feasible integral schedule that is 2-approximate against the optimal $\lpp$ cost. We guarantee the feasibility by constructing a continuous schedule that completes each job $j$ by time $C^*_j$ such that $\sum_j w_j C^*_j$ is at most twice the optimal $\lpp$ cost. In Section~\ref{sec:rounding2}, we find an integral schedule meeting the deadlines, which will be our final schedule. 

\subsection{Reduction to Deadline Constrained Scheduling}
	\label{sec:rounding1}
	
We first give a high-level overview of the first step of the rounding procedure that we will present in this section. The procedure considers the following four schedules in this order. 

\begin{enumerate}
	\item {\bf $\{y_j^F\}$: An optimal solution to $\lpp$.} We view this as a sequence of fractional matchings over \emph{integer} time slots. The cost is measured as the solution's LP objective, which is denoted as $\cost(\lpp)$.
	\item {\bf $\sigma$: A continuous schedule constructed from $\{y_j^F\}$.} The cost is measured as the schedule's total weighted \emph{fractional completion time}: job $j$ is alive at time $t$ by at most $(1 - v)$ fraction and incurs a cost of at most $(1 - v)$ at the time if $j$'s \emph{each} demand has been served by at least $v$ fraction by the time.  The total cost over all jobs is denoted as $\cost(\sigma)$ and we will show $\cost(\sigma) = \cost(\lpp) - (1/2) \sum_j w_j$.
	\item {\bf $\sigma^\lambda$: A continuous schedule obtained by `stretching' $\sigma$.} We stretch the schedule $\sigma$ horizontally by a factor of $1 / \lambda$. $\lambda \in (0,1)$ is to be chosen randomly but the random choice can be derandomized. The cost is measured \emph{integrally}, meaning that $j$ completes when \emph{all} its demands are served. We will show 
	$\mathbb{E}[\cost(\sigma^\lambda)] \leq 2\ \cost(\sigma) = 2\ \cost(\lpp) - \sum_j w_j$.
	\item {\bf $\{C^*_j\}$: A continuous schedule with integer deadlines.} This schedule is essentially the same as $\sigma^\lambda$. In $\sigma^\lambda$, jobs may have fractional completion times. We round them up to their nearest integers, respectively. Then, the cost of any schedule that completes each job $j$ before its predefined completion time (deadline) is at most $\mathbb{E}[\sum_j w_j C^*_j] \leq   \mathbb{E}[\cost(\sigma^\lambda)] + \sum_j w_j  	\leq 2\ \cost(\lpp)$. 
\end{enumerate}

The second step of our rounding process will find an actual schedule that complete each job $j$ by $C^*_j$, which will be discussed in Section~\ref{sec:rounding2} in detail. 

In the rest of this section, we show how we construct each of the above schedules and upper bound its cost. 

\paragraph{The First Schedule $(\{y_j^F\})$:} Let $\{y_j^F\}$ be an optimal  solution to $\lpp$.
We use these $y_j^F$ values to construct a sequence of fractional matchings over integer time slots. 
Let $\displaystyle x_{io}^{jt} = \sum_{F \in \mathcal{F}(j) : (i,o,t) \in F} y_j^F$ denote the fraction of data of job $j$ that is sent from port $i \in [m]$ to port $o \in [m]$ in time slot $t \in [T] \setminus [r_j]$. Since we started from a feasible LP solution, the fractional coflow schedule defined by the $x_{io}^{jt}$ variables schedules $d_{io}^j$ amount of data for any pair $(i,o)$ of ports.
\begin{align*}
\sum_{t \in [T]} x_{io}^{jt} &= \sum_{t \in [T]} \sum_{\substack{F \in \mathcal{F}(j)\\(i,o,t) \in F}} y_j^F
= \sum_{F \in \mathcal{F}(j)} \sum_{\substack{t \in [T]}} y_j^F \mathbbm{1}_{(i,o,t) \in F}
= \sum_{F \in \mathcal{F}(j)} y_j^F d_{io}^j = d_{io}^j
\end{align*}
Also for any fixed time slot $t$, the variables $x_{io}^{jt}$ describe a fractional matching, i.e.
\begin{align*}
\sum_{o \in [m]} \sum_{j \in J} x_{io}^{jt} &= \sum_{j \in J} \sum_{o \in [m]} \sum_{\substack{F \in \mathcal{F}(j)\\(i,o,t) \in F}} y_j^F \overset{\eqref{eqn:LPinputconflict}}{\leq} 1 \quad,\forall i \in [m]\\
\textnormal{and } \sum_{i \in [m]} \sum_{j \in J} x_{io}^{jt} &= \sum_{j \in J} \sum_{i \in [m]} \sum_{\substack{F \in \mathcal{F}(j)\\(i,o,t) \in F}} y_j^F \overset{\eqref{eqn:LPoutputconflict}}{\leq} 1 \quad,\forall o \in [m]
\end{align*}

Note that the optimal cost of $\lpp$ can be written as 
\begin{equation*}
	\cost(\lpp) :=  \sum_j w_j \sum_{F \in \cF(j)} C_j^F y_j^F  =\sum_j w_j \sum_{C \geq 1} C \cdot \sum_{F \in \cF(j): C_j^F = C} y_j^F
\end{equation*}

\paragraph{The Second Schedule ($\sigma$):} We construct a continuous schedule $\sigma$ from $\{y_j^F\}$. Recall that in a continuous schedule, an integral matching is scheduled at each instantaneous time. Consider any fixed time slot $t \in [T]$. Let $X^t$ denote the fractional matching given by $\{y_j^F\}$. By the Birkhoff-von Neumann theorem, $X^t$ can be decomposed into a polynomial number of integral matchings $M^t_1, M^t_2, \ldots, M^t_\ell$ such that
\[X^t = \sum_{i=1}^\ell \alpha^t_i M^t_i \textnormal{\quad and \quad} \sum_{i=1}^\ell \alpha^t_i = 1\]
Thus, we can convert the fractional schedule for time slot $t$ into a continuous time schedule by 
scheduling each integral matching $M^t_i$ for $\alpha_i (\tau_2 - \tau_1)$ time during each infinitesimal time interval $[\tau_1, \tau_2) \in [t-1, t)$. In other words, the linear combination of integral matchings is `smeared' across all instantaneous time during $[t-1, t)$.
Let $\sigma$ denote such a continuous schedule obtained by decomposing the fractional matching in each time slot $t \in [T]$. For any time interval $[\tau_1, \tau_2)$, let $\int_{\tau = \tau_1}^{\tau_2} \sigma_{io,j}(\tau) d\tau$ denote the amount of data of co-flow $j$ that is sent from port $i$ to port $o$ in the schedule $\sigma$. By construction, we have that
\begin{align*}
\int_{\tau=t-1}^t \sigma_{io,j}(\tau) d\tau = x_{io}^{jt}. 
\end{align*}

To compute $\cost(\sigma)$, we need to measure how much fraction of each job $j$ is completed by time $\tau$. Recall that in this measure, job $j$ is said to be completed by $v$ fraction at time $\tau$, if for every pair $i,o$ of ports, at least $v d_{io}^j$ amount of data is transferred from $i$ to $o$ by time $\tau$ in schedule $\sigma$, i.e., $\int_{\tau' = 0}^\tau \sigma_{io,j}(\tau') d \tau' \geq v d^j_{io}$. Let $\tilde C_j(v)$ denote the first time when $j$ is completed by $v$ fraction in $\sigma$. The total fractional completion time of job $j$ in schedule $\sigma$ is thus defined as $\int_{v =0}^1 \tilde C_j(v) d v$.
Formally, we now want to upper bound $\cost(\sigma) := \sum_j w_j \int_{v =0}^1 \tilde C_j(v) d v$.

\begin{lemma}
\label{lem:sigmacost}
$\cost(\sigma) = \cost(\lpp) - \frac{1}{2}(\sum_{j \in J} w_j)$.
\end{lemma}
\begin{proof}
Consider a fixed job $j$ and a configuration $F \in \cF(j)$  with $y_j^F > 0$. Let $t = C_j^F$ denote the last time slot during which $F$ schedules a matching. Let us compute the contribution of job $j$ to $\cost(\sigma)$ due to configuration $F$.

Let $M$ denote the matching scheduled by $F$ in this last time slot. The LP solution schedules matching $M$ by a $y_j^F$ fraction in time slot $t$. Since we `smeared' the fractional matching in time slot $t$ over continuous times in the interval $[t-1, t)$, our continuous schedule $\sigma$ schedules matching $M$ by a $y_j^F d\tau$ fraction during any time interval $[\tau, \tau + d \tau) \in [t-1, t)$. Thus $j$'s contribution to $\cost(\sigma)$ due to configuration $F \in \cF(j)$ is $w_j\int_{\tau = t-1}^t y_j^F \tau d\tau = w_jy_j^F(C_j^F - 1/2)$.

Summing over all jobs $j$ and all configurations, we have
\begin{align*}
\cost(\sigma) &= \sum_j w_j \int_{v =  0}^{1} \tilde C_j(v) d v = \sum_j w_j \sum_{F \in \cF(j)} y_j^F(C_j^F- 1/2) = \cost(\lpp) - \sum_j w_j /2 \label{sched2}
\end{align*}
as desired.
\end{proof}

Finally, we note that we don't have to construct $\sigma$ explicitly, as we just need to find $\tilde C_j(v)$ in polynomial time for given $j$ and $v \in [0, 1]$.

\paragraph{The Third Schedule ($\sigma^\lambda$):} We construct a new continuous schedule $\sigma^\lambda$ from $\sigma$ as follows. First, we choose a $\lambda \in [0,1]$ randomly drawn according to the probability density function $f(v) = 2v$. Next, we ``stretch" the schedule $\sigma$ by a factor of $1/\lambda$. More precisely, if matching $M$ is scheduled in $\sigma$ during an infinitesimal interval $[\tau_1, \tau_2)$, the same matching is scheduled in $\sigma^\lambda$ during $[\tau_1 / \lambda, \tau_2 / \lambda)$. 
This `stretching' (also called slow-motion) idea  was used in other scheduling contexts  \cite{schulz1997random,queyranne20022+,im2014preemptive}.
Note that jobs' release times remain respected after the stretching.

Let $C_{\sigma^\lambda}(j)$ denote the completion time of job $j$ in the stretched schedule, i.e., $C_{\sigma^\lambda}(j)$ is the earliest time such that $d_{io}^j$ amount of data has been transferred from port $i$ to port $o$ for all pairs $i,o \in [m]$.
We now measure the cost of this new schedule, $\cost(\sigma^\lambda) := \sum_j w_j C_{\sigma^\lambda}(j)$, and upper bound it.

\begin{lemma}
\label{lem:stretching}
$\mathbb{E}[\cost(\sigma^\lambda)] = 2\ \cost(\sigma).$
\end{lemma}

\begin{proof}
Recall that we define $\tilde C_j(v)$ to be the earliest time when job $j$ is completed by $v$ fraction in $\sigma$. It is easy to observe that for every job $j$, $C_{\sigma^\lambda}(j) \leq \tilde C_j(\lambda) / \lambda$. This is because, $\sigma$ schedules at least a $\lambda$ fraction  of the demand for each pair $i,o$ by time $\tilde C_j(\lambda)$ and thus $\sigma^\lambda$ schedules all the demands fully by time $\tilde C_j(\lambda) / \lambda$.
\begin{align*}
\mathbb{E}[\cost(\sigma^\lambda)] &= \sum_j w_j \mathbb{E}[C_{\sigma^\lambda}(j)] \leq \sum_j w_j \mathbb{E}[ \tilde C_j(\lambda) / \lambda] = \sum_j w_j \int_{v=0}^1 \tilde C_j(v) / v \cdot (2v) dv\\
&\leq 2 \sum_j w_j \int_{v=0}^1 \tilde C_j(v) dv = 2\  \cost(\sigma)
\end{align*}
\end{proof}

\paragraph{The Last Schedule ($\{C^*_j\}$) :} In the previous continuous schedule, $\sigma^\lambda$, every job $j$ is completed by time $C_{\sigma^\lambda}(j)$, which is not necessarily an integer. Since we will use this completion time as $j$'s deadline in the second rounding procedure to find an actual schedule, we need to make sure the completion times / deadlines are integers. Hence, we set $C_j := \lceil C_{\sigma^\lambda}(j) \rceil$. 
Then, we immediately have,
\begin{equation}
	\label{sched4}
\cost(\{C^*_j\}) = \sum_j w_j C^*_j \leq \cost(\sigma^\lambda) + \sum_j w_j.
\end{equation}

From Lemma \ref{lem:sigmacost}, Lemma \ref{lem:stretching} and Eq.\ \eqref{sched4}, we have $\mathbb{E}[\cost(\{C^*_j\})] \leq 2\ \cost(\lpp)$. Thus, we have found integer deadlines, $\{C_j^*\}$, such that there exists a continuous schedule that completes each job $j$ by time $C_j^*$ and $\sum_j w_j C_j^*$ is at most twice the optimal $\lpp$ cost in expectation.

\medskip
We note that the random choice of $\lambda$ can be derandomized. To keep the flow of our presentation, we defer the derandomization to Section~\ref{sec:derandomization}.

\subsection{Finding a Feasible Integral Schedule Meeting Deadlines}
	\label{sec:rounding2}

This section is devoted to finding an integral schedule that complete all jobs before their respective deadlines, $\{C_j^*\}$. Recall that $\sigma^\lambda$ is such a schedule, except that it is continuous, not integral. We first convert the continuous time schedule $\sigma^\lambda$ back into a time-slotted fractional schedule such that each co-flow $j$ is completed by time $C_j^*$. Define $z_{io}^{jt} = \int_{\tau = t-1}^t \sigma_{io,j}^\lambda(\tau) d\tau$ for all $t \leq \lceil\tau_{io,j}^\lambda\rceil$. It is easy to verify that the set $\{z_{io}^{jt}\}$ of values define a feasible fractional coflow schedule such that each coflow finishes by time $C^*_j$. In particular, we have the following. 
\begin{align}
\sum_{t \leq C_j^*} z_{io}^{jt} = \sum_{t \leq C_j^*} \int_{\tau = t-1}^t \sigma_{io,j}^\lambda(\tau) d\tau = \int_{\tau = 0}^{\tau_{io,j}^\lambda} \sigma_{io,j}^\lambda(\tau) d\tau = d_{io}^j,\quad\quad\mbox{$\forall j \in J$ and $\forall i,o \in [m]$} \label{eq:zdemand}\\
\sum_{j \in J} \sum_{o \in [m]} z_{io}^{jt} = \sum_{j \in J} \sum_{o \in [m]} \int_{\tau = t-1}^t \sigma_{io,j}^\lambda(\tau) d\tau \leq 1,\quad\quad\mbox{$\forall t \in [T]$ and $\forall i \in [m]$}
\label{eq:zinputcap}\\
\sum_{j \in J} \sum_{i \in [m]} z_{io}^{jt} = \sum_{j \in J} \sum_{i \in [m]} \int_{\tau = t-1}^t \sigma_{io,j}^\lambda(\tau) d\tau \leq 1,\quad\quad\mbox{$\forall t \in [T]$ and $\forall o \in [m]$}
\label{eq:zoutputcap}
\end{align}
The last inequality in the above two statements follows from construction as $\sigma^\lambda$ schedules an integral matching at any instantaneous time $\tau$.

As a final step, we construct an instance of the network flow problem with integer capacities such that integral, maximum flows in this network are in one-to-one correspondence with integral feasible coflow schedules. Our final integral coflow schedule is obtained by solving a maximum flow problem on this network. This reduction to network flow is similar to the one used for the separation oracle in Section \ref{sec:oracle} and we highlight the differences below.

Construct a directed graph $G$ as follows. For each $t \in [T]$, let $G_t = (I_t, O_t, A_t)$ be a complete, bipartite, directed subgraph with unit capacity edges and vertices.
For each job $j$ and pair $i,o \in [m]$, add vertex $source_{io,j}$ and connect it to all vertices $i_t$ for all $t \in [C^*_j] \setminus [r_j]$. Similarly, add vertex $sink_{io,j}$ and connect it to all vertices $o_t$ in the same range.
Finally, add two new vertices $source$ and $sink$. For every coflow $j$ and pair $i,o \in [m]$, add an arc of capacity $d_{io}^j$ from $source$ to $source_{io,j}$ and also from $sink_{io,j}$ to $sink$. 
In contrast with the construction for the separation oracle that only considered a fixed co-flow $j$, here we add source and sink vertices for all co-flows.

Note that any integral maximum flow $f$ in $G$ of value $\sum_{j \in J} \sum_{i,o \in [m]} d_{io}^j$ corresponds to a feasible schedule $F$ such that any coflow $j$ has completion time at most $C^*_j$ and vice versa. In particular, for any coflow $j$ and pair $i,o \in [m]$, the maximum flow sends $d_{io}^j$ units of flow from $source_{io,j}$ to $sink_{io,j}$ through $d_{io}^j$ edge disjoint paths. The vertex capacities ensure that for any $t \in [T]$, the edges of $G_t$ carrying flow form a matching. 

Finally we demonstrate that the $z_{io}^{jt}$ values defined above yield a fractional flow $f$ of value $\sum_{j \in J} \sum_{i,o \in [m]} d_{io}^j$. Since all capacities in this network are integers, this guarantees the existence of an integral flow of the same value that can be found using an off-the-shelf network flow algorithm. Let $f(e)$ denote the flow through arc $e$. We define $f$ as follows.
\begin{align}
f(e) = 
\begin{cases}
\sum_j z_{io}^{jt}, &\quad\quad\mbox{for $e = (i_t,o_t) \in A_t$} \\
z_{io}^{jt}, &\quad\quad\mbox{for $e = (source_{io,j}, i_t)$ and $e = (o_t, sink_{io,j})$}\\
d_{io}^j, &\quad\quad\mbox{for $e = (source, source_{io,j})$ and $e = (sink_{io,j}, sink)$}
\end{cases}
\end{align}
Equations \eqref{eq:zinputcap} and \eqref{eq:zoutputcap} guarantee that the flow $f$ satisfies all edge and vertex capacity constraints. Equation \eqref{eq:zdemand} ensures that the value of the flow is indeed $\sum_{j \in J} \sum_{i,o \in [m]} d_{io}^j$ as desired.

\section{Removing the Simplifying Assumption: Handling Large Job Demands and Release Times}
	\label{sec:remove}

In this section, we remove the simplifying assumption stated in Section~\ref{sec:simplifying} at an extra $(1+\eps)$ factor loss in the approximation ratio. As the assumption was used in several places, we discuss how to remove the assumption from each place. The key ideas remain the same; thus, we only discuss the major differences in detail.

\subsection{LP Formulation}

We only need to consider times no greater than $T := \max_{j \in J} r_j + \sum_{j \in J, i \in [m], o \in [m]} d_{io}^j$. This is because we can finish all jobs by time $T$ even if we deliver only one packet at a time after the last job's arrival and there is no reason to idle all ports at any time as long as there is a packet that is ready for schedule. As $T$ could be exponentially large, we use a standard trick that partitions the whole set of integer times, $\{1, 2, 3, \cdots, T\}$ into disjoint integer time intervals of exponentially increasing lengths, $I_1$, $I_2$, ..., $I_\kappa$ such that for any two integers, $t_1 < t_2$ in the same interval, we have $t_2 / t_1 \leq 1 + \eps$. This can be done, for example, by setting $I_1 = \{1 \}$, $I_2 = \{2\}$, $\cdots$, $I_{10 /\eps} = \{\frac{10}{\eps}\}$, 
and $I_{10/\eps+1} = \{\lceil \frac{10}{\eps} (1+\eps)^0 \rceil + 1, \ldots, \lceil \frac{10}{\eps} (1+\eps)^1 \rceil\}$, $I_{10 / \eps+2} = \{\lceil \frac{10}{\eps} (1+\eps)^1 \rceil + 1, \ldots, \lceil \frac{10}{\eps} (1+\eps)^2 \rceil\}$, and so on, where $\kappa = O((1 / \eps) \log T)$.
Further, we recursively split an interval if the interval contains the release time of some job $j$ until every interval has at most one job's release time, and if so, it is the smallest time in the interval. Let $\cI$ denote the collection of the created intervals. For notational convenience, we don't change the notation, so that $I_1, I_2, \cdots, I_\kappa \in \cI$ are ordered in increasing order of the times in them. 

To compactly describe configurations, we do not distinguish times in same interval $I_k$. Intuitively, as any two times in the same interval differ by a factor of at most $1+\eps$, this compact description will increase the optimal LP objective by a factor of at most $1+\eps$. For a job $j$, a configuration is a complete,  feasible, integral sketch for $j$. Formally, a configuration $F$ for job $j$ is a ternary relation of tuples of the form $(i,o,I)$, which is associated with a weight $f(i,o,I)$, that indicates $f(i,o,I)$ units of data are transferred from input port $i$ to output port $o$ during interval $I \in \cI$. For a feasible configuration $F$ for job $j$, the relation must satisfy:

\begin{enumerate}
	\item  $\sum_{(i,o,I) \in F} f(i,o,I)  = d^j_{io}$ for all $i, o \in [m]$.
	\item  $\sum_{(i,o,I) \in F} f(i,o,I) \leq |I|$ for every pair of $i \in [m]$ and $I \in \cI$. 
	\item  $\sum_{(i,o,I) \in F} f(i,o,I) \leq |I|$ for every pair of $o \in [m]$ and $I \in \cI$. 
\end{enumerate}

Let $\mathcal{F}(j)$ denote the set of all possible feasible configurations of $j$. Let $C_j^F$ denote the completion time of job $j$ under configuration $F \in \mathcal{F}(j)$. More precisely, $C_j^F := \max \{t \ | \ \exists i, o \in [m], I \in \cI \mbox{ s.t. } f(i,o,I) > 0, t \in I\}$. In a similar spirit, define $C(I) := \max \{ t \in I\}$. We only need to change two constraints (\ref{eqn:LPinputconflict}) and (\ref{eqn:LPoutputconflict}) of $\lpp$ as follows.

\begin{enumerate}
  \item $\forall i \in [m] \textnormal{ and } \forall I \in \cI$, $\displaystyle \sum_{j \in J, o \in [m]} \sum_{\substack{F \in \mathcal{F}(j)   s.t. (i,o,I) \in F}} y_j^F \cdot f(i,o,I) \leq |I|$.
  \item $\forall o \in [m] \textnormal{ and } \forall I \in \cI$, $\displaystyle \sum_{j \in J, i \in [m]} \sum_{\substack{F \in \mathcal{F}(j) s.t. (i,o,I) \in F}} y_j^F \cdot f(i,o,I) \leq |I|$.
 \end{enumerate}

This new $\lpp$'s optimal objective is at most $(1+\eps)$ off the optimum, since times in the same interval differ by a factor of at most $(1+\eps)$ and if a job completes during an interval, we pretend that it does at the last time in the interval. Note that since $\cI$ has a polynomial number of intervals, the new $\lpp$ only has a polynomial number of constraints. 

\subsection{Solving the LP} 

The dual constraint (\ref{eqn:LPdual})  of $\lpd$ becomes:
\begin{align*}
  & \forall j \in J \textnormal{ and } \forall F \in \mathcal{F}(j), && \displaystyle \alpha_j - \sum_{(i,o,I) \in F}  (\beta_{i,I} + \gamma_{o,I}) \cdot f(i, o, I) \leq w_j C_j^F  
\end{align*}

Fix a job $j$ and $I_k \in \cI$. The key problem for the separation oracle becomes the following. During interval $I$, input port $i$ and output port $o$ are priced at $\beta_{i,I}$ and $\gamma_{o,I}$, respectively, per unit traffic. Our task is to find a minimum-cost configuration for $J$ that only uses $I_1, I_2, ..., I_k \in \cI$. As before, we reduce the problem to an instance of the minimum cost flow problem. For each interval $I \in \{I_1, I_2, ..., I_k\}$, instead of each $t$, we create a complete, bipartite, directed subgraph $G_I = (I_I, O_I, A_I)$. The cost of edge $(i_I, o_I)$ is set to $\beta_{i,I} + \gamma_{o,I}$.
All unit capacities in $G_t$ are increased to $|I|$. Similarly, $source_{io}$ is connected to vertices $i_I$ for every interval $I \in \{I_1, I_2, ..., I_k\}$ that does not include any time smaller than $r_j$, with capacity $|I|$. Likewise, the capacities of the edges connecting vertices $o_I$ to $sink_{io}$ are increased to $|I|$. Note that any integral maximum flow $f$ in $G$ corresponds to a feasible configuration $F$ for job $j$ with completion time at most $C(I_k)$ and vice versa. Further, the cost of the flow sending $f'(i,o,I)$ units of flow from $i$ to $o$ during interval $I$ is exactly $\sum_{i,o} (\beta_{i,I} + \gamma_{o,I}) \cdot f'(i,o,I)$. Thus, we can obtain a feasible configuration $F \in \mathcal{F}(j)$ with completion time at most $C(I_k)$ that minimizes $\sum_{(i,o,t) \in F} (\beta_{i,t} + \gamma_{o,t}) \cdot  f(i,o,I)$ from a minimum cost maximum flow in $G$.

\subsection{First Rounding}

In this section, we adapt the rounding in Section~\ref{sec:rounding1} to avoid using the simplifying assumption. Since the entire rounding is almost identical, we only describe the key differences. For each $I = \{t_1, t_2, ..., t_h\} \in \cI$, from an optimal solution to $\lpp$, we construct a continuous schedule that schedules the same integral matchings uniformly over the interval $[t_1 - 1, t_h)$. Then, we can still show $\cost(\sigma) \leq \cost(\lpp) - (1/2) \sum_j w_j$---we can get a tighter bound but this bound suffices for our goal. The remaining procedure is exactly the same, so we can show $\mathbb{E}[\sum_j w_j C^*_j] \leq 2\ \cost(\lpp)$.

\subsection{Second Rounding}

The second rounding also remains almost identical. There are two main differences that are worth special attention. First, we need to refine $\cI$ to include the discovered deadlines $C^*_j$. This is because we treat all times in the same interval in $\cI$ equally, and as a result, we can't strictly enforce jobs' deadlines which could potentially yield an integral schedule whose cost is more than twice the optimal LP cost. Therefore, we recursively refine intervals, so that each interval $I \in \cI$ contains at most one distinct deadline $C^*_j$, and if so, it is the largest time in the interval. As before, we do not distinguish times in the same interval though. Thus, we will
use $I$ in place of $t$ and increase the capacity of any edge involving interval $I$ from one to  $|I|$. 

The other key difference is this. Take a close look at the discovered integral maximum flow for a fixed interval $I \in \cI$. Observe that in the flow, each node supports at most $|I|$ units of flow. In other words, unit sized packets are delivered from input ports to output ports to (partially) serve some jobs and each port is used by at most $|I|$ packets. To obtain a feasible integral schedule, we convert this into $|I|$ integral matchings. As mentioned earlier, it was already observed in \cite{qiu2015minimizing} that such a conversion is possible: a bipartite graph can be decomposed into a set of integral matchings whose number is equal to the graph's maximum vertex degree. Thus, we can find in polynomial time an integral schedule that is 2-approximate against the optimal $\lpp$ cost. This implies the discovered schedule is $2(1+\eps)$-approximation, as the relaxation is at most $(1+\eps)$ factor off the optimum.

\section{Derandomization}
	\label{sec:derandomization}
In this section, we discuss how to derandomize the random choice of $\lambda \in (0,1]$, which was used in the first rounding to construct a stretched schedule $\sigma^\lambda$ from $\sigma$. Let us first define \emph{step} values. We say that $v \in (0, 1]$ is a step value if $\sum_{F \in \cF(j): C^F_j \leq C} y^F_j= v$ for some $j \in J$ and integer $C \geq 1$ --- in other words,  exactly $v$ fraction of some job $j$'s configurations are completed by some integer time in the $\lpp$ solution, $\{y^F_j\}$. Let $V$ denote the set of all step values; $1 \in V$ by definition. We can assume w.l.o.g. that $|V|$ is polynomially bounded, as we can find an optimal solution to $\lpp$ with a polynomial-size support.

We have shown in Section~\ref{sec:rounding1}
 that $\mathbb{E}[\cost(\sigma^\lambda)] \leq 2\cost(\lpp) - \sum_j w_j$ and $\cost(\{C_j^*\}) \leq \cost(\sigma^\lambda) + \sum_j w_j$. Therefore, if we find $\lambda$ such that $\cost(\sigma^\lambda) \leq 2\cost(\lpp) - \sum_j w_j$ (it is easy to see such a $\lambda$ exists from a simple averaging argument), then we can find $\{C^*_j\}$ such that $\cost(\{C_j^*\}) \leq 2 \cost(\lpp)$. 
 
 Our remaining goal is to find  $\lambda$ that minimizes $\cost(\sigma^\lambda)$. Towards this end, fix a job $j$ and take a close look at $j$'s contribution to $\cost(\sigma^\lambda)$ depending on the choice of $\lambda$. Fix any two adjacent step values $v_1 < v_2$ in $V$. Suppose $\lambda$ was set to a value $v \in (v_1, v_2]$.
Let $t$ be the earliest time slot such that $\sum_{F \in \cF(j): C^F_j \leq t} y^F_j \geq v_2$.  
Let $v_0 := \sum_{F \in \cF(j): C^F_j \leq t-1} y^F_j$; note that $v_0 \leq v_1$, as $v_1$ and $v_2$ are adjacent step values. Then, we have $C_{\sigma^\lambda}(j) =\frac{1}{v} (t - 1 + \frac{v - v_0}{v_2 - v_0})$.
This is because exactly $v_0$ fraction of $j$'s configurations complete by time $t-1$ and there is $(v_2 - v_0)$ fraction of $j$'s configurations completing at time $t$; thus, $\tilde C_j(v) = t - 1 + \frac{v - v_0}{v_2 - v_0}$. Hence, for any $v \in (v_1, v_2]$, $j$ contributes to $\cost(\sigma^\lambda)$ by $w_j \frac{1}{v} (t - 1 + \frac{v - v_0}{v_2 - v_0})$. This becomes a linear function in $z$ over $[1/v_2,1/v_1)$ if we set $z = 1/v$. Therefore, we get a piece-wise linear function $g(z)$ by summing over all jobs and considering all pairs of two adjacent step values in $V$. We set $\lambda$ to the the inverse of $z$'s value that achieves the global minimum, which can be found in polynomial time.
 
\section{Extensions To Non-Uniform Capacity Constraints}
\label{sec:extensions}

In this section, we consider the co-flow scheduling problem with non-uniform port capacities. The problem setting remains exactly the same as regular co-flow scheduling, with the exception that port $i$ has an integer capacity $c_i$ that denotes the number of packets that can be processed by port $i$ at any time $t$. We show that our algorithm can be easily extended to give the same approximation guarantee even for this generalized problem. Because the key ideas remain the same, we only highlight the primary differences in each step.

\subsection{LP Formulation and Separation Oracle}

Let $\cF(j)$ denote the set of all feasible configurations of job $j$. Note that a configuration $F \in \cF(j)$ may schedule up to $c_i$ units of data transfer from (or to) port $i$ at time slot $t$. For an input port $i$, let $load_F(i,t)$ denote the number of packets transferred from port $i$ at time $t$ as per configuration $F$. We similarly define $load_F(o,t)$ for an output port $o$. We only need to change constraints \eqref{eqn:LPinputconflict} and \eqref{eqn:LPoutputconflict} as follows.
\begin{align*}
& \forall i \in [m] \textnormal{ and } \forall t \in [T], && \displaystyle \sum_{j \in J} \sum_{F \in \mathcal{F}(j)} load_F(i,t)\ y_j^F \leq c_i \\
  & \forall o \in [m] \textnormal{ and } \forall t \in [T], && \displaystyle \sum_{j \in J} \sum_{F \in \mathcal{F}(j)} load_F(o,t)\ y_j^F \leq c_o
\end{align*}

The separation oracle for the dual LP once again relies on finding the minimum cost schedule for a job $j$ subject to a fixed deadline. Similar to Section \ref{sec:oracle}, this key problem can be solved using minimum cost flow. Since the reduction is almost identical to the one earlier, we skip the details.

\subsection{Rounding}
The key difference in the rounding step occurs in the construction of the continuous schedule $\sigma$. The fractional schedule obtained from the LP relaxation no longer schedules a fractional matching during a time slot $t$. Let $X^t$ be the bipartite graph scheduled by the feasible LP solution at time slot $t$. We are guaranteed that the fractional degree of each vertex $i \in X^t$ is at most $c_i$. To construct the continuous schedule, we need to decompose $X^t$ into a convex combination of a polynomial number of degree bounded subgraphs $B_1, B_2, \ldots, B_\ell$. Such a decomposition can be obtained by splitting a vertex $i \in X^t$ into $c_i$ distinct copies and splitting the edge weights so that each vertex has at most unit fractional degree in the expanded graph. The Birkhoff-von Neumann decomposition of the expanded graph then gives us the desired decomposition.

Once we have the continuous time schedule $\sigma$, the rest of the rounding procedure remains identical to that in Sections \ref{sec:rounding1} and Sections \ref{sec:rounding2}.

\bibliographystyle{plain}
\bibliography{coflow}

\begin{thebibliography}{10}

\bibitem{ahmadi2017scheduling}
Saba Ahmadi, Samir Khuller, Manish Purohit, and Sheng Yang.
\newblock On scheduling coflows.
\newblock In {\em IPCO}, pages 13--24. Springer, 2017.

\bibitem{bansal2010inapproximability}
Nikhil Bansal and Subhash Khot.
\newblock Inapproximability of hypergraph vertex cover and applications to
  scheduling problems.
\newblock In {\em ICALP}, pages 250--261. Springer, 2010.

\bibitem{chen2007supply}
Zhi-Long Chen and Nicholas~G Hall.
\newblock Supply chain scheduling: Conflict and cooperation in assembly
  systems.
\newblock {\em Operations Research}, 55(6):1072--1089, 2007.

\bibitem{chowdhury2012coflow}
Mosharaf Chowdhury and Ion Stoica.
\newblock Coflow: A networking abstraction for cluster applications.
\newblock In {\em ACM Workshop on Hot Topics in Networks}, pages 31--36. ACM,
  2012.

\bibitem{chowdhury2015efficient}
Mosharaf Chowdhury and Ion Stoica.
\newblock Efficient coflow scheduling without prior knowledge.
\newblock In {\em SIGCOMM}, pages 393--406. ACM, 2015.

\bibitem{chowdhury2014efficient}
Mosharaf Chowdhury, Yuan Zhong, and Ion Stoica.
\newblock Efficient coflow scheduling with varys.
\newblock In {\em SIGCOMM}, SIGCOMM '14, pages 443--454, New York, NY, USA,
  2014. ACM.

\bibitem{dean2008mapreduce}
Jeffrey Dean and Sanjay Ghemawat.
\newblock Mapreduce: simplified data processing on large clusters.
\newblock {\em Communications of the ACM}, 51(1):107--113, 2008.

\bibitem{garg2007order}
Naveen Garg, Amit Kumar, and Vinayaka Pandit.
\newblock Order scheduling models: Hardness and algorithms.
\newblock In {\em FSTTCS}, pages 96--107. Springer, 2007.

\bibitem{im2014preemptive}
Sungjin Im, Maxim Sviridenko, and Ruben Van Der~Zwaan.
\newblock Preemptive and non-preemptive generalized min sum set cover.
\newblock {\em Mathematical Programming}, 145(1-2):377--401, 2014.

\bibitem{JahanjouKR16}
Hamidreza Jahanjou, Erez Kantor, and Rajmohan Rajaraman.
\newblock Asymptotically optimal approximation algorithms for coflow
  scheduling.
\newblock {\em CoRR}, abs/1606.06183, 2016.

\bibitem{khuller2017select}
Samir Khuller, Jingling Li, Pascal Sturmfels, Kevin Sun, and Prayaag Venkat.
\newblock Select and permute: An improved online framework for scheduling to
  minimize weighted completion time.
\newblock {\em arXiv preprint arXiv:1704.06677}, 2017.

\bibitem{khuller2016brief}
Samir Khuller and Manish Purohit.
\newblock Brief announcement: Improved approximation algorithms for scheduling
  co-flows.
\newblock In {\em Proceedings of the 28th ACM Symposium on Parallelism in
  Algorithms and Architectures}, pages 239--240. ACM, 2016.

\bibitem{leung2007scheduling}
Joseph Y-T Leung, Haibing Li, and Michael Pinedo.
\newblock Scheduling orders for multiple product types to minimize total
  weighted completion time.
\newblock {\em Discrete Applied Mathematics}, 155(8):945--970, 2007.

\bibitem{luo2016towards}
S.~Luo, H.~Yu, Y.~Zhao, S.~Wang, S.~Yu, and L.~Li.
\newblock Towards practical and near-optimal coflow scheduling for data center
  networks.
\newblock {\em IEEE Transactions on Parallel and Distributed Systems},
  PP(99):1--1, 2016.

\bibitem{mastrolilli2010minimizing}
Monaldo Mastrolilli, Maurice Queyranne, Andreas~S Schulz, Ola Svensson, and
  Nelson~A Uhan.
\newblock Minimizing the sum of weighted completion times in a concurrent open
  shop.
\newblock {\em Operations Research Letters}, 38(5):390--395, 2010.

\bibitem{qiu2015minimizing}
Zhen Qiu, Cliff Stein, and Yuan Zhong.
\newblock Minimizing the total weighted completion time of coflows in
  datacenter networks.
\newblock In {\em SPAA}, SPAA '15, pages 294--303, New York, NY, USA, 2015.
  ACM.

\bibitem{queyranne20022+}
Maurice Queyranne and Maxim Sviridenko.
\newblock A (2+ $\varepsilon$)-approximation algorithm for the generalized
  preemptive open shop problem with minsum objective.
\newblock {\em Journal of Algorithms}, 45(2):202--212, 2002.

\bibitem{sachdeva2013optimal}
Sushant Sachdeva and Rishi Saket.
\newblock Optimal inapproximability for scheduling problems via structural
  hardness for hypergraph vertex cover.
\newblock In {\em IEEE Conference on Computational Complexity}, pages 219--229.
  IEEE, 2013.

\bibitem{schulz1997random}
Andreas~S Schulz and Martin Skutella.
\newblock Random-based scheduling new approximations and lp lower bounds.
\newblock In {\em International Workshop on Randomization and Approximation
  Techniques in Computer Science}, pages 119--133. Springer, 1997.

\bibitem{shafiee2017improved}
Mehrnoosh Shafiee and Javad Ghaderi.
\newblock An improved bound for minimizing the total weighted completion time
  of coflows in datacenters.
\newblock {\em arXiv preprint arXiv:1704.08357}, 2017.

\bibitem{wang2007customer}
Guoqing Wang and TC~Edwin Cheng.
\newblock Customer order scheduling to minimize total weighted completion time.
\newblock {\em Omega}, 35(5):623--626, 2007.

\bibitem{zaharia2010spark}
Matei Zaharia, Mosharaf Chowdhury, Michael~J Franklin, Scott Shenker, and Ion
  Stoica.
\newblock Spark: Cluster computing with working sets.
\newblock {\em HotCloud}, 10(10-10):95, 2010.

\bibitem{zhao2015rapier}
Yangming Zhao, Kai Chen, Wei Bai, Minlan Yu, Chen Tian, Yanhui Geng, Yiming
  Zhang, Dan Li, and Sheng Wang.
\newblock Rapier: Integrating routing and scheduling for coflow-aware data
  center networks.
\newblock In {\em INFOCOM}, pages 424--432. IEEE, 2015.

\end{thebibliography}

\newpage

\appendix

\end{document}